\documentclass{article}
\usepackage{url}
\usepackage{amsmath}
\usepackage{amsthm}
\usepackage{amssymb}
\usepackage{amsfonts}
\usepackage{stmaryrd}
\usepackage{enumerate}
\usepackage{multicol}
\usepackage[lined,boxed]{algorithm2e}

\usepackage{tikz}
\usetikzlibrary{arrows}

\usetikzlibrary{arrows,positioning} 
\tikzset{
    >=stealth',
    punkt/.style={
           text width=13em,
           minimum height=2em,
           text centered},
    pil/.style={
           ->,
           thick,
           shorten <=2pt,
           shorten >=2pt,}
}

          \newcommand{\EXPTIME}{\textnormal{\textnormal{\textsc{Exptime}}}}
          \newcommand{\T}{\ensuremath{\mathcal{T}}}
          \newcommand{\AF}[1]{\ensuremath{\llbracket {#1} \rrbracket}}

          \newtheorem{Theorem}{Theorem}[section]
          \newtheorem{Proposition}[Theorem]{Proposition}
          \newtheorem{Observation}[Theorem]{Observation}
          \newtheorem{Fact}[Theorem]{Fact}
          \newtheorem{Lemma}[Theorem]{Lemma}
          \newtheorem{Remark}[Theorem]{Remark}
          \newtheorem{Corollary}[Theorem]{Corollary}
          \newtheorem*{Lemma*}{Lemma}

          \theoremstyle{definition}
          
          \newtheorem*{Def*}{Definition}
          \newtheorem{Example}[Theorem]{Example}

          \DeclareMathOperator{\MSO}{\textnormal{MSO}}
          \DeclareMathOperator{\FO}{\textnormal{FO}}
          \DeclareMathOperator{\mDatalog}{\textnormal{mDatalog}}
          \DeclareMathOperator{\PiMSO}{\text{\normalsize$\forall$\hspace{1pt}\scriptsize$\exists$}\textnormal{-MSO}}

          \DeclareMathOperator{\Rel}{\textnormal{\textbf{rel}}}

          \DeclareMathOperator{\Root}{\textnormal{\textbf{root}}}

          \DeclareMathOperator{\Label}{\textnormal{\textbf{label}}}
          \DeclareMathOperator{\Child}{\textnormal{\textbf{child}}}
          \DeclareMathOperator{\Is}{\textnormal{\textbf{are\_siblings}}}
          \DeclareMathOperator{\As}{\Is}

          \DeclareMathOperator{\Desc}{\textnormal{\textbf{desc}}}
          \DeclareMathOperator{\Fc}{\textnormal{\textbf{firstchild}}}
          \DeclareMathOperator{\Ns}{\textnormal{\textbf{nextsibling}}}
          \DeclareMathOperator{\Ls}{\textnormal{\textbf{lastsibling}}}
          \DeclareMathOperator{\Leaf}{\textnormal{\textbf{leaf}}}

          \DeclareMathOperator{\idb}{\textnormal{idb}}
          \DeclareMathOperator{\edb}{\textnormal{edb}}
          \newcommand*{\Var}{\ensuremath{\textit{var}}}

          \newsavebox{\probbox}  
          \newlength{\problength}  
          \newenvironment{Problem}[1]{        
            \begin{center}  
            \setlength{\problength}{0.95\textwidth} 
            \begin{lrbox}{\probbox}        
            \begin{minipage}{\problength}   
               \textnormal{\textnormal{#1}}     
               \vspace{1ex}
                \begin{quote}    
                \begin{description}  
                }{     
                \end{description}         
                  \end{quote}  
            \end{minipage}\end{lrbox}\noindent\fbox{\usebox{\probbox}}   
            \end{center}    
           }        
          \newcommand{\In}{\item[\textnormal{\textit{Input:}}]}  
          \newcommand{\Out}{\item[\textnormal{\textit{Output:}}]}  
          
          \newcommand{\Quest}{\item[\textnormal{\textit{Question:}}]}

          \newcommand{\Yes}{\ensuremath{\textnormal{\texttt{Yes}}}}
          \newcommand{\No}{\ensuremath{\textnormal{\texttt{No}}}}

          \newcommand{\markEnd} {\hfill\raisebox{-1pt}{$ \lrcorner$}}
          \renewcommand{\phi}{\varphi}

          \newenvironment{mi}{\begin{itemize}}{\end{itemize}}
          
          \newenvironment{mea}{\begin{enumerate}[(a)]}{\end{enumerate}}

          \newcommand{\nc}{\newcommand}
          \nc{\rnc}{\renewcommand}

          \rnc{\geq}{\ensuremath{\geqslant}}
          \rnc{\leq}{\ensuremath{\leqslant}}

          \nc{\deff}{:=}

          \nc{\set}[1]{\ensuremath{\{ #1 \}}}
          \nc{\setc}[2]{\set{#1 \,:\, #2}}

          \nc{\isom}{\ensuremath{\cong}}

          \nc{\dist}{\ensuremath{\textit{dist}}}

          \nc{\ov}[1]{\ensuremath{\overline{#1}}}
          \nc{\mymod}{\ensuremath{\textup{ mod }}}

          \nc{\NN}{\ensuremath{\mathbb{N}}}
          \nc{\Z}{\ensuremath{\mathbb{Z}}} 
          \nc{\NNpos}{\ensuremath{\NN_{\mbox{\tiny $\scriptscriptstyle \geq 1$}}}}

          \nc{\und}{\ensuremath{\wedge}}
          \nc{\Und}{\ensuremath{\bigwedge}}
          \nc{\oder}{\ensuremath{\vee}}
          \nc{\Oder}{\ensuremath{\bigvee}}
          \nc{\nicht}{\ensuremath{\neg}}
          \nc{\impl}{\ensuremath{\rightarrow}}
          \nc{\gdw}{\ensuremath{\leftrightarrow}}

          \nc{\free}{\ensuremath{\textit{free}}}
          \nc{\ar}{\ensuremath{\textit{ar}}} 
          \nc{\qr}{\ensuremath{\textit{qr}}}
          \nc{\size}[1]{\ensuremath{|\!|#1|\!|}}

          \nc{\A}{\ensuremath{\mathcal{A}}}
          \nc{\B}{\ensuremath{\mathcal{B}}}
          \nc{\C}{\ensuremath{\mathcal{C}}}
          \nc{\PP}{\ensuremath{\mathcal{P}}}
          \rnc{\S}{\ensuremath{\mathcal{S}}}

          \nc{\SGK}{\ensuremath{\S_{\textit{GK}}}}
          \nc{\tauGK}{\ensuremath{\tau_{\textit{GK}}}}

           \nc{\Class}{\ensuremath{\mathfrak{C}}}

          \nc{\atoms}{\ensuremath{\textit{atoms}}}

          \nc{\rroot}{\ensuremath{\textit{root}}}

\nc{\IDB}[1]{\ensuremath{\textit{#1}}}
\nc{\ANS}{\IDB{Ans}}
\nc{\WHITE}{\IDB{White}}

\nc{\todo}[1]{\ \\\framebox{\parbox{\textwidth}{TODO: #1}}}
\nc{\Nicole}[1]{\marginpar{{\scriptsize\raggedright\textbf{N: }#1}}}
\nc{\Andre}[1]{\marginpar{\scriptsize\raggedright\textbf{A: }#1}}
\rnc{\todo}[1]{}
\rnc{\Nicole}[1]{}
\rnc{\Andre}[1]{}

\begin{document}
\title{A note on monadic datalog on unranked trees 
}
\author{Andr\'e Frochaux \and Nicole Schweikardt}
\date{Goethe-Universit\"at Frankfurt am Main\\[0.5ex] 
  {\ttfamily
   \{afrochaux|schweika\}@informatik.uni-frankfurt.de}
 \\[4ex]
 Version: \today}
\maketitle
\begin{abstract}
In the article 
\emph{Recursive queries on trees and data trees} (ICDT'13), 
Abiteboul \emph{et al.}\ \,asked whether the containment
problem for monadic datalog over unordered unranked labeled trees
using the child relation and the descendant relation is
decidable. This note gives a positive answer to this
question, as well as an overview of the relative
expressive power of monadic datalog on various representations of
unranked trees.
\end{abstract}

\section{Introduction}\label{section:Introduction}

The logic and database theory literature has considered various kinds of
representations of finite labeled trees as logical structures. In
particular, trees are either ranked or unranked (i.e., the number of
children of each node is bounded by a constant, or unbounded);
the children of each node are either ordered or unordered; and
there is or there is not available the descendant relation
(i.e., the transitive closure of the child relation);
for overviews see \cite{DBLP:journals/lmcs/Libkin06,
DBLP:conf/csl/Neven02,
WThomas-Handbook-survey,
tata2008}.

Considering \emph{ordered} unranked labeled trees,
Gottlob and Koch \cite{GottlobKoch} showed that monadic datalog,
viewed as a language for
defining Boolean or unary queries on such trees, is
exactly as expressive as monadic second-order logic. For achieving
this result, they represent a tree as a logical structure where the
nodes of the tree form the structure's universe, on which there are
available the firstchild relation, the nextsibling relation, and
unary relations for representing the root, the leaves, the last
siblings, and the labels of the nodes.
Other papers, e.g.\
\cite{DBLP:journals/tcs/NevenS02,DBLP:journals/jacm/BojanczykMSS09},
consider representations of trees where also the child relation and
its transitive closure, the descendant relation are available.

For \emph{unordered} unranked labeled trees, one usually
considers logical representations consisting only of the child relation, and
possibly also
the descendant relation, along with unary relations for encoding the
node labels, cf.\ e.g.\ \cite{AbiteboulBMW13,Kepser2008,DBLP:journals/jcss/BjorklundMS11}.
Recently, Abiteboul \emph{et al.} \cite{AbiteboulBMW13} considered recursive query languages on
unordered trees and data trees, among them datalog and monadic
datalog. In particular, they asked for the decidability of the query
containment problem for monadic datalog on unordered labeled trees
represented using the child relation and the descendant relation.
The present paper gives an affirmative answer to this question, as
well as an overview of the expressive power of monadic datalog on
various representations of trees as logical structures.

The paper is organised as follows. 
Section~\ref{section:Preliminaries} fixes the basic notation concerning
unordered as well as ordered trees, and their representations as
logical structures. Furthermore, it recalls the syntax and semantics,
along with basic properties, of monadic datalog and monadic
second-order logic.
Section~\ref{section:ExpressivePower} gives details on the expressive
power of monadic datalog on various kinds of tree representations.
Section~\ref{section:QCPofMonadicDatalog} shows that query
containment, equivalence, and satisfiability of monadic
datalog queries are decidable on all considered tree representations.

\section{Preliminaries}\label{section:Preliminaries}

We write $\NN$ for the set of non-negative integers, and we let
$\NNpos \deff \NN\setminus\set{0}$.
For a set $S$ we write $2^S$ to denote the power set of $S$, i.e., the
set $\setc{X}{X\subseteq S}$.
\\
Throughout this paper, we let $\Sigma$ be a fixed finite non-empty
alphabet.

\subsection{Relational Structures}\label{subsection:structures}

In this paper, a \emph{schema} (or, \emph{signature}) $\tau$ consists of a
finite number of relation symbols $R$, each of a fixed
\emph{arity} $\ar(R)\in\NNpos$.
A \emph{$\tau$-structure} $\A$ consists of a \emph{finite} non-empty
set $A$ called the \emph{domain} (or, \emph{universe}) of $\A$, and a relation
$R^{\A}\subseteq A^{\ar(R)}$ for each relation symbol $R\in\tau$. 
Sometimes, it will be convenient to
identify $\A$ with the \emph{set of atomic facts of $\A$},
i.e., the set
\begin{eqnarray*}
   \atoms(\A) 
 & \deff 
 & \setc{\ R(a_1,\ldots,a_r)\ }{\ \ R\in\tau,\ \ r=\ar(R),\ \ 
   (a_1,\ldots,a_r)\in R^\A \ }.
\end{eqnarray*}

If $\tau$ and $\tau'$ are schemas such that $\tau\subseteq \tau'$,
and $\A$ is a $\tau$-structure and $\B$ a $\tau'$-structure, then
$\A$ is the \emph{$\tau$-reduct} of $\B$
(and $\B$ is a \emph{$\tau'$-expansion} of $\A$), if $\A$ and $\B$
have the same domain and $R^\A=R^\B$ is true for all $R\in\tau$.

\subsection{Unordered Trees}\label{subsection:UnorderedTrees} 

An \emph{unordered $\Sigma$-labeled tree} $T=(V^T,\lambda^T,E^T)$
consists of a finite set $V^T$ of nodes, a function  
$\lambda^T: V^T\to \Sigma$ assigning to each node $v$ of $T$ a label
$\lambda(v)\in\Sigma$, and a set $E^T\subseteq V^T\times V^T$ of directed edges such
that the following is true: 
\begin{mi}
 \item 
  There is exactly one node $\rroot^T\in V^T$ with in-degree 0. This node is
  called the \emph{root} of $T$.
 \item
  Every node $v\in V^T$ with $v\neq \rroot^T$ has in-degree 1, and
  there is exactly one directed path from $\rroot^T$ to $v$.
\end{mi}

\noindent
As in \cite{AbiteboulBMW13}, we represent unordered
$\Sigma$-labeled trees $T$ by relational structures $\S_u(T)$ of schema
\begin{eqnarray*}
 \tau_u & \deff & 
 \setc{\,\Label_\alpha}{\alpha\in\Sigma\,}
 \; \cup \;
 \set{\,\Child\,}, 
\end{eqnarray*}
where $\Child$ has arity~2 and $\Label_\alpha$ has arity~1 (for every $\alpha\in\Sigma$), as follows:
\begin{mi}
 \item
   The domain of $\S_u(T)$ is the set $V^T$ of all nodes of $T$,
 \item
   for each label $\alpha\in \Sigma$, 
   $\Label_\alpha^{\S_u(T)}$ consists of all nodes labeled $\alpha$, i.e.
   $\Label_\alpha^{\S_u(T)} = \setc{v\in V^T}{\lambda^T(v)=\alpha}$,
   and
 \item
   $\Child^{\S_u(T)} = E^T$.
\end{mi}

\begin{figure}[!h]
\begin{center}
\begin{tikzpicture}[->,>=stealth']
\tikzset{
  treenode/.style = {align=center, inner sep=0pt, text centered,
    font=\sffamily},
  arn_red/.style = {treenode, circle, white, font=\sffamily\bfseries, draw=black,
    fill=black, text width=1.5em},
  arn_blue/.style = {treenode, circle, black, draw=black, 
    text width=1.5em, very thick},
}

  \node [arn_red] {$v_0$}
    child {node [arn_red] {$v_1$}}
    child  {node [arn_blue] {$v_2$}
			child  {node [arn_blue] {$v_6$}}    				
			child  {node [arn_red] {$v_7$}}
    				}
    child  {node [arn_red] {$v_3$}}
    child  {node [arn_blue] {$v_4$}
    		child  {node [arn_red] {$v_8$}}}
    child  {node [arn_red] {$v_5$}}
    ;
\end{tikzpicture}
\caption{An example tree $T$ labeled by symbols from $\Sigma=\{\textit{Black},\textit{White}\}$.} \label{Pic:tree}
\end{center}
\end{figure}
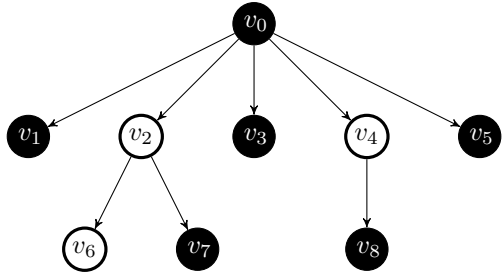

\begin{Example} \label{Exa:tree-unordered}
Let $T$ be the unordered\footnote{Note that an unordered tree does not contain any information on the
relative order of the children of a node. Thus, the arrangement of
children given in the picture is only one of many possibilities to draw the tree.}
$\Sigma$-labeled tree from
Figure~\ref{Pic:tree}, for $\Sigma=\{\textit{Black},\textit{White}\}$.
The $\tau_u$-structure $\A= \S_u(T)$ representing $T$ has domain
\begin{eqnarray*}
   A & = & \set{v_0,v_1,v_2,v_3,v_4,v_5,v_6,v_7,v_8}
\end{eqnarray*}
and relations
\begin{mi}
 \item $\Label_{\textit{Black}}^\A = \set{v_0,v_1,v_3,v_5,v_7,v_8}$,
 \item $\Label_{\textit{White}}^\A = \set{v_2,v_4,v_6}$,
 \item $\Child^{\A} = \left\{
     \begin{array}{l} 
       (v_0,v_1), (v_0,v_2),(v_0,v_3),(v_0,v_4),(v_0,v_5), \\
       (v_2,v_6),(v_2,v_7),(v_4,v_8) 
     \end{array} \right\}.$
\end{mi}
The set of atomic facts of $\A$ is the set $\atoms(\A)=$
\[
 \left\{
  \begin{array}{l} 
    \Label_{\textit{Black}}(v_0), \ 
    \Label_{\textit{Black}}(v_1), \ 
    \Label_{\textit{Black}}(v_3), \ 
    \Label_{\textit{Black}}(v_5), \\ 
    \Label_{\textit{Black}}(v_7), \
    \Label_{\textit{Black}}(v_8), \
    \Label_{\textit{White}}(v_2), \ 
    \Label_{\textit{White}}(v_4), \\ 
    \Label_{\textit{White}}(v_6), \ 
    \Child(v_0,v_1), \ 
    \Child(v_0,v_2), \
    \Child(v_0,v_3), \\
    \Child(v_0,v_4), \
    \Child(v_0,v_5), \
    \Child(v_2,v_6), \
    \Child(v_2,v_7), \
    \Child(v_4,v_8) 
  \end{array}
 \right\}.
\]
\markEnd
\end{Example}

\medskip

\noindent
Sometimes, we will also consider the extended schema
\begin{eqnarray}
  \tau'_u & \deff & \tau_u \ \cup \ \set{\,\Desc, \ \Is, \ \Root, \ \Leaf \,},
\end{eqnarray}
where $\Desc$ and $\Is$ are of arity 2, and $\Root$ and $\Leaf$ are
of arity 1.
\\
The $\tau'_u$-representation $\S'_u(T)$ of an unordered
$\Sigma$-labeled tree $T$ is the expansion of $\S_u(T)$ by the
relations
\begin{mi}
 \item
   $\Desc^{\S'_u(T)}$, which is the transitive (and non-reflexive) closure
   of $E^T$,
 \item 
   $\Is^{\S'_u(T)}$, which consists of all tuples $(u,v)$ of nodes
   such that $u\neq v$ have the same parent (i.e., 
       there is a $w\in V^T$ such that $(w,u)\in E^T$ and
       $(w,v)\in E^T$).
 \item 
   $\Root^{\S'_u(T)}$ consists of the root node $\rroot^T$ of $T$,
 \item
   $\Leaf^{\S'_u(T)}$ consists of all leaves of $T$, i.e., all $v\in
   V^T$ that have out-degree~0 w.r.t.\
     $E^T$.
\end{mi}

\noindent
For a set $M\subseteq\set{\Desc,\Is,\Root,\Leaf}$ we let 
\begin{eqnarray*}
 \tau_u^M & \deff &
 \tau_u \cup M,
\end{eqnarray*}
and for every $\Sigma$-labeled unordered tree $T$ we let 
$\S^M_u(T)$ be the $\tau_u^M$-reduct of $\S'_u(T)$. 
If $M$ is a singleton set, we omit the curly brackets --- in particular, we 
write $\tau_u^{\Desc}$ instead of
$\tau_u^{\set{\Desc}}$, and  $\S_u^{\Desc}(T)$ instead of $\S_u^{\set{\Desc}}(T)$.

\subsection{Ordered Trees}\label{subsection:OrderedTrees} 

An \emph{ordered} $\Sigma$-labeled tree $T=(V^T,\lambda^T,E^T,\textit{order}^T)$
consists of the same components as an unordered $\Sigma$-labeled tree
and, in addition, $\textit{order}^T$ fixes, for each node $u$ of $T$,
a strict linear order of all the children\footnote{i.e., the
nodes $v$ such that $(u,v)\in E^T$} of $u$ in $T$.

We represent ordered $\Sigma$-labeled trees $T$ by relational
structures $\S_o(T)$ of schema
\begin{eqnarray*}
 \tau_o & \deff &
 \setc{\,\Label_\alpha}{\alpha\in \Sigma\,} \ \cup \ \set{\,\Fc,\Ns\,},
\end{eqnarray*}
where $\Fc$ and $\Ns$ have arity 2 and $\Label_\alpha$ has arity 1
(for every $\alpha\in\Sigma$) as follows:
\begin{mi}
 \item The domain of $\S_o(T)$ is the set $V^T$ of all nodes of $T$,
 \item for each $\alpha\in\Sigma$, the relation $\Label_\alpha^{\S_o(T)}$
   is defined in the same way as for unordered trees,
 \item $\Fc^{\S_o(T)}$ consists of all tuples $(u,v)$ of nodes such
   that $u$ is the first child of $v$ in $T$ (i.e., 
   $\textit{order}^T$ lists $u$ as the first child of $v$),
 \item $\Ns^{\S_o(T)}$ consists of all tuples $(v,v')$ of nodes such
   that $v$ and $v'$ have the same parent, i.e., there is an
   $u\in V^T$ such that $(u,v)\in E^T$ and $(u,v')\in E^T$, and
   $v'$ is the immediate successor of $v$ in the linear order of the
   children of $u$ given by $\textit{order}^T$.
\end{mi}
Often, we will also consider the extended schema
\begin{eqnarray*}
 \tau'_o & \deff &
 \tau_o \ \cup \ \set{\,\Child,\ \Desc,\ \Root,\ \Leaf,\ \Ls \,},
\end{eqnarray*}
where $\Child$ and $\Desc$ have arity 2 and $\Root$, $\Leaf$, $\Ls$ have arity 1.
The $\tau'_o$-representation $\S'_o(T)$ of an ordered $\Sigma$-labeled
tree $T$ is the expansion of $\S_o(T)$ by the relations
\begin{mi}
 \item 
   $\Child^{\S'_o(T)}$, $\Desc^{\S'_o(T)}$, $\Root^{\S'_o(T)}$, and $\Leaf^{\S'_o(T)}$, which are
   defined in the same way as for unordered trees, and
 \item $\Ls^{\S'_o(T)}$, which consists of all nodes $v\neq \rroot^T$ such
   that $\textit{order}^T$ lists $v$ as the last child of its parent $u$.
\end{mi}
For a set $M\subseteq\set{\Child,\ \Desc,\ \Root,\ \Leaf,\ \Ls}$ we let
\begin{eqnarray*}
 \tau_o^M & \deff & \tau_o\cup M,
\end{eqnarray*}
and for every $\Sigma$-labeled ordered tree $T$ we let $\S_o^M(T)$ be
the $\tau_o^M$-reduct of $\S'_u(T)$. If $M$ is a singleton set, we
omit curly brackets.

Note that in \cite{GottlobKoch}, Gottlob and Koch represented ordered $\Sigma$-labeled trees 
$T$ by relational structures $\SGK(T)\deff \S_o^{\set{\Root,\Leaf,\Ls}}$ of schema
\[
 \begin{array}{c}
 \tauGK \ \ \deff \ \ \tau_o^{\set{\Root,\Leaf,\Ls}} \ \ =  \ \
 \tau'_o\setminus\set{\Child,\,\Desc} \ \ = \ \ 
 \\[2ex]
 \setc{\,\Label_\alpha}{\alpha\in \Sigma\,} \ \cup \
 \{\, \Fc,\ \Ns,\ \Root,\ \Leaf,\ \Ls \,\}.
\end{array}
\]

\medskip

\begin{Example}\label{Exa:tree-ordered}
Let $T$ be the \emph{ordered} $\Sigma$-labeled tree from
Figure~\ref{Pic:tree}, for $\Sigma=\{Black,White\}$, where the order
of the children of each node is from left to right, as depicted in the illustration.
The $\tauGK$-structure $\B=\SGK(T)$ representing $T$ has
domain
\begin{eqnarray*}
   B & = & \set{v_0,v_1,v_2,v_3,v_4,v_5,v_6,v_7,v_8}
\end{eqnarray*}
and relations
\begin{mi}
 \item $\Label_{\textit{Black}}^\B = \set{v_0,v_1,v_3,v_5,v_7,v_8}$,
 \item $\Label_{\textit{White}}^\B = \set{v_2,v_4,v_6}$,
 \item $\Root^\B = \set{v_0}$,
 \item $\Leaf^\B = \set{v_1,v_3,v_5,v_6,v_7,v_8}$,
 \item $\Fc^\B = \set{\,(v_0,v_1), \ (v_2,v_6),\ (v_4,v_8)\,}$,
 \item $\Ns^\B = \set{\,(v_1,v_2),\ (v_2,v_3),\ (v_3,v_4),\ (v_4,v_5),\ (v_6,v_7)\,}$,
 \item $\Ls^\B = \set{v_5,v_7,v_8}$.
\end{mi}
Note that the root node of $T$ is not included in any sibling relation.
\markEnd
\end{Example}

\subsection{Monadic Datalog ($\mDatalog$)}\label{subsection:datalog}

The following definition of monadic datalog ($\mDatalog$, for short)
is basically taken from \cite{GottlobKoch}.

 A \textit{datalog rule} is an expression of the form 
\[ 
  h \leftarrow  b_1 ,\ldots, b_n ,
\]
for $n\in\NN$,
where $h, b_1 , \ldots, b_n$ are called \emph{atoms} of the rule, $h$
is called the rule's \emph{head}, and $b_1,\ldots,b_n$ (understood as a
conjunction of atoms) is called the \emph{body}.
Each atom is of the form $P(x_1,\ldots,x_m)$ where $P$ is a predicate
of some arity $m\in\NNpos$ and $x_1,\ldots,x_m$ are variables.
Rules are required to be \emph{safe} in the sense that all variables
appearing in the head also have to appear in the body.

A \textit{datalog program} is a finite set of datalog rules. 
Let $\PP$ be a datalog program and let $r$ be a datalog rule. 
We write 
$\Var(r)$ for the set of all variables occurring in the rule $r$, and
we let $\Var(\PP):= \bigcup_{r \in \PP} \Var(r)$. 
Predicates that occur in the head of
some rule of $\PP$ are called \emph{intensional}, whereas predicates that only
occur in the body of rules of $\PP$ are called \emph{extensional}.
We write $\idb(\PP)$ and
$\edb(\PP)$ to denote the sets of intensional and extensional predicates
of $\PP$, and we say that $\PP$ \emph{is of schema $\tau$}
if $\edb(\PP)\subseteq \tau$.
A datalog program belongs to \emph{monadic datalog} ($\mDatalog$, for
short), if all its \emph{intensional} predicates have arity~1. 

For defining the semantics of datalog, let $\tau$ be a schema,
let $\PP$ be a datalog program of schema $\tau$, let
$A$ be a domain, and let
\begin{eqnarray*}
 F_{\PP,A} & \deff & 
 \setc{\ R(a_1,\ldots,a_r)}{R\in \tau\cup\idb(\PP),\ r=\ar(R),\ a_1,\ldots,a_r\in A
 \ }
\end{eqnarray*}
the set of all \emph{atomic facts over $A$}. 
A \textit{valuation $\beta$ for $P$ in $A$} is a function $\beta: \big(\Var(P)
\cup A\big) \to A$ where $\beta(a)=a$ for all $a\in A$.
For an atom $P(x_1,\ldots,x_m)$ occurring in a rule of $\PP$ we
let 
\begin{eqnarray*}
   \beta\big(P(x_1,\ldots,x_m)\big)
 & \deff & 
   P\big(\beta(x_1),\ldots,\beta(x_m)\big) \ \ \in \ F_{\PP,A}.
\end{eqnarray*}    
The \emph{immediate consequence operator} $\T_{\PP}:2^{F_{\PP,A}}\to
2^{F_{\PP,A}}$ induced by the datalog program $\PP$ on domain $A$ maps
every $C\subseteq F_{\PP,A}$ to 
\begin{eqnarray*}
 \T_{\PP}(C) & \deff &
  C \ \cup \
 \left\{\ 
    \beta(h) \ : 
    \begin{array}{p{6cm}}
     there is 
     a rule $h \leftarrow b_1,\ldots,b_n$ in $\PP$ and
     a valuation $\beta$ for $\PP$ in $A$ such that 
     $\beta(b_1),\ldots,\beta(b_n)\in C$
  \end{array}
 \right\}.
\end{eqnarray*}
Clearly, $\T_{\PP}$ is \emph{monotone}, i.e., for $C\subseteq D\subseteq
F_{\PP,A}$ we have $\T_{\PP}(C)\subseteq \T_{\PP}(D)$.

Letting $\T_{\PP}^0(C)\deff C$ and $\T_{\PP}^{i+1}(C)\deff
\T_{\PP}\big(\T_{\PP}^i(C)\big)$ for all $i\in\NN$, it is straightforward to
see that
\[
  C = \T_{\PP}^0(C) \subseteq \T_{\PP}^1(C) \subseteq \cdots \subseteq
  \T_{\PP}^{i}(C) \subseteq \T_{\PP}^{i+1}(C) \subseteq \cdots \subseteq F_{\PP,A}.
\]
For a finite domain $A$, the set $F_{\PP,A}$ is finite, and hence
there is an $i_0\in\NN$ such that $\T_{\PP}^{i_0}(C)=\T_{\PP}^i(C)$ for all
$i\geq i_0$. In particular, the set $\T_{\PP}^\omega(C)\deff\T_{\PP}^{i_0}(C)$ is a
\emph{fixpoint} of the immediate consequence operator $\T_{\PP}$.
By the theorem of Knaster and Tarski we know that this fixpoint is the
\emph{smallest} fixpoint of $\T_{\PP}$ which contains $C$.

\begin{Theorem}[Knaster and Tarski \cite{Tar}]\label{thm:KnasterTarski}
Let $\tau$ be a schema, let $\PP$ be a datalog program of schema $\tau$,
and let $A$ be a finite domain.
For every $C\subseteq F_{\PP,A}$ we have
\begin{eqnarray*}
   \T_{\PP}^{\omega}(C) 
 & = 
 & \bigcap \; \setc{\,D}{\,\T_{\PP}(D) = D \text{ \,and \,} C
   \subseteq D \subseteq F_{\PP,A}\,} \\
 & = 
 & \bigcap \; \setc{\,D}{\,\T_{\PP}(D) \subseteq D \text{ \,and \,} C
   \subseteq D\subseteq F_{\PP,A}\,}.
\end{eqnarray*} 
\markEnd
\end{Theorem}

A \emph{$k$-ary (monadic) datalog query} of schema $\tau$ is a tuple $Q=(\PP,P)$
where $\PP$ is a (monadic) datalog program of schema $\tau$ and $P$ is an
(intensional or extensional) predicate of arity $k$ occurring in
$\PP$.
$\PP$ and $P$ are called the \emph{program} and the \emph{query
  predicate} of $Q$.
When evaluated in a finite $\tau$-structure $\A$, the query $Q$
results in the following $k$-ary relation over $A$:
\begin{eqnarray*}
  \AF{Q}(\A) & \deff &
  \setc{\ (a_1,\ldots,a_k)\in A^k\ }{\ \; P(a_1,\ldots,a_k) \;\in\;
    \T_\PP^\omega\big(\atoms(\A)\big)\ }.
\end{eqnarray*}
\emph{Unary} queries are queries of arity $k=1$.

The \emph{size} $\size{Q}$ of a (monadic) datalog query $Q$ is the
length of $Q=(\PP,P)$ viewed as a string over alphabet
\[
 \edb(\PP) \, \cup \, \idb(\PP) \, \cup \,
 \set{x,y,z,{}_{0},{}_{1},{}_{2},{}_{3},{}_{4},{}_{5},{}_{6},{}_{7},{}_{8},{}_{9}}
 \, \cup \, 
 \set{(,),\{,\}} \, \cup \, \set{\leftarrow}\cup\set{,}.
\]

\begin{Example} \label{sample:2red}
Consider the schema $\tauGK$ introduced in Section~\ref{subsection:OrderedTrees} for
representing ordered $\Sigma$-labeled trees for
$\Sigma=\set{\textit{Black},\ \textit{White}}$.
We present a unary monadic datalog query
$Q=(\PP,\ANS)$ of schema $\tauGK$ such that for every ordered
$\Sigma$-labeled tree $T$ we have
\begin{eqnarray*}
  \AF{Q}\big(\SGK(T)\big) 
 & = &
  \left\{
   \begin{array}{lp{6cm}}
     \set{\,\rroot^T\,} & if the root of $T$ has exactly two
       children labeled with the symbol \textit{White},
    \\[1ex]
     \emptyset & otherwise.
   \end{array}
  \right.
\end{eqnarray*}
To this end, we let $\PP$ consist of the following rules:
\[
\begin{array}{rcl}
\ANS(x)     &\leftarrow & \Root(x),\ \Fc(x,y),\ \WHITE_2(y)\\[0.5ex]
\WHITE_2(x) &\leftarrow &  \Label_{Black}(x),\ \Ns(x,y),\ \WHITE_2(y)\\[0.5ex]
\WHITE_2(x) &\leftarrow &  \Label_{White}(x),\ \Ns(x,y),\ \WHITE_1(y)\\[0.5ex]
\WHITE_1(x) &\leftarrow &  \Label_{Black}(x),\ \Ns(x,y),\ \WHITE_1(y)\\[0.5ex]
\WHITE_1(x) &\leftarrow &  \Label_{White}(x),\ \Ns(x,y),\ \WHITE_0(y)\\[0.5ex]
\WHITE_0(x) &\leftarrow &  \Label_{Black}(x),\ \Ns(x,y),\ \WHITE_0(y)\\[0.5ex]
\WHITE_1(x) &\leftarrow &  \Label_{White}(x),\ \Ls(x) \\[0.5ex]
\WHITE_0(x) &\leftarrow &  \Label_{Black}(x),\ \Ls(x)
\end{array}
\]
In particular, $Q$ returns $\set{\rroot^T}$ on the tree from
Example~\ref{Exa:tree-ordered}. \markEnd 
\end{Example}

\begin{Remark}[Folklore]\label{remark:monotonicity}
The monotonicity of the immediate consequence operator implies
that datalog queries $Q$ of schema $\tau$ are monotone in the following sense:
If $\A$ and $\B$ are $\tau$-structures with $\atoms(\A)\subseteq
\atoms(\B)$, then $\AF{Q}(\A)\subseteq \AF{Q}(\B)$. \markEnd
\end{Remark}

Let us point out that it is also well-known that datalog is \emph{preserved
under homomorphisms} in the following sense.
A \emph{homomorphism} from a $\tau$-structure $\A$ to a
$\tau$-structure $\B$ is a mapping $h:A\to B$ such that for all
$R\in\tau$ and all tuples $(a_1,\ldots,a_r)\in R^\A$ we have
$(h(a_1),\ldots,h(a_r))\in R^\B$. 
As a shorthand, for any set $S\subseteq A^k$ we
let $h(S)=\setc{\big(h(a_1),\ldots,h(a_k)\big)}{(a_1,\ldots,a_k)\in S}$.

\begin{Lemma}[Folklore]\label{lemma:homomorphisms}
Any $k$-ary datalog query $Q$ of schema $\tau$ is preserved under
homomorphisms in the following sense: If $\A$ and $\B$ are
$\tau$-structures, and $h$ is a homomorphism from $\A$ to $\B$, then
$h\big(\AF{Q}(\A)\big)\subseteq \AF{Q}(\B)$.
\end{Lemma}
\begin{proof}
Let $\A$ and $\B$ be $\tau$-structures and let $h:A\to B$ be a
homomorphism from $\A$ to $\B$. Furthermore, let $Q=(\PP,P)$ where
$\PP$ is a datalog program of schema $\tau$. 
For an atomic fact $f=R(a_1,\ldots,a_r)\in F_{\PP,A}$ let
$h(f)\deff R(h(a_1),\ldots,h(a_r))$ be the according atomic fact in
$F_{\PP,B}$.
Furthermore, for a set $S\subseteq F_{\PP,A}$ let
$h(S)\deff\setc{h(f)}{f\in S}$ be the according subset of $F_{\PP,B}$.

First, note that by the definition of the immediate consequence
operator $\T_\PP$ it is straightforward to see that the following is true:
If $C\subseteq F_{\PP,A}$ and $D\subseteq F_{\PP,B}$ such that $h(C)\subseteq D$,
then $h(\T_\PP(C))\subseteq \T_\PP(D)$.

Next, note that this immediately implies that the following is true: If
$C\subseteq F_{\PP,A}$ and $D\subseteq F_{\PP,B}$ such that
$h(C)\subseteq D$, then $h\big(\T_\PP^\omega(C)\big)\subseteq \T_\PP^\omega(D)$.

Finally, note that $h$ is a homomorphism from $\A$ to $\B$, and thus 
$h\big(\atoms(\A)\big)\subseteq \atoms(\B)$. Consequently, 
$h\big(\T_\PP^\omega\big(\atoms(\A)\big)\big)\subseteq
\T_\PP^\omega\big(\atoms(\B)\big)$. In particular, this means that
$h\big(\AF{Q}(\A)\big)\subseteq \AF{Q}(\B)$.
\end{proof}

\subsection{Monadic Second-Order Logic ($\MSO$)}

The set $\MSO(\tau)$ of all monadic second-order formulas of schema
$\tau$ is defined as usual, cf.\ e.g.\ \cite{Libkin}: There are
two kinds of variables, namely 
\emph{node variables}, denoted with lower-case letters $x$, $y$, $\ldots$,
$x_1$, $x_2$, $\ldots$ and ranging over elements of the domain,
and \emph{set variables}, denoted with upper-case letters $X$, $Y$, $\ldots$,
$X_1$, $X_2$, $\ldots$ and ranging over sets of elements of the domain. 

\noindent
An \emph{atomic} $\MSO(\tau)$-formula is of the form
\begin{description}
	\item[(A1)] $R(x_1,\ldots,x_r)$, \ where $R\in\tau$, $r=\ar(R)$,
          and $x_1,\ldots,x_r$ are node variables,
	\item[(A2)] $x=y$, \ where $x$ and $y$ are node variables, \ or
	\item[(A3)] $X(x)$, \ where $x$ is a node variable and $X$ is a set variable.
\end{description}

\noindent
If $x$ is a node variable, $X$ a set variable, and $\phi$ and $\psi$
are $\MSO(\tau)$-formulas, then 
\begin{description}
	\item[(BC)] $\lnot \phi$ \ and \ $(\phi \lor \psi)$ \ are $\MSO(\tau)$-formulas,
	\item[(Q1)] $\exists x \phi$ \ and \ $\forall x \phi$ \ are $\MSO(\tau)$-formulas,
	\item[(Q2)] $\exists X \phi$ \ and \ $\forall X \phi$ \ are $\MSO(\tau)$-formulas.
\end{description}

Quantifiers of the form (Q1) are called \emph{first-order
  quantifiers}; quantifiers of the form (Q2) are called \emph{set
  quantifiers}.
$\MSO(\tau)$-formulas in which no set quantifier occurs, are called
\emph{first-order formulas}
(\emph{$\FO(\tau)$-formulas}, for short).
The \emph{size} $\size{\varphi}$ of a formula $\varphi$ is the length
of $\varphi$ viewed as a string over alphabet
\[
 \tau
 \, \cup \,
 \set{x,y,z,X,Y,Z,{}_{0},{}_{1},{}_{2},{}_{3},{}_{4},{}_{5},{}_{6},{}_{7},{}_{8},{}_{9}}
 \, \cup \, 
 \set{(,)} \, \cup \, \set{=,\nicht,\oder,\exists,\forall}\cup\set{,}.
\]

As shortcuts we use the Boolean connectives $(\varphi\und\psi)$, $(\varphi\impl\psi)$, and
$(\varphi\gdw\psi)$, the statement $x\neq y$ for node
variables, and the statements $X=Y$, $X\neq Y$, and
$X\subseteq Y$ for set variables. Note that all these can easily be
expressed in first-order logic.
To improve readability of formulas, we will sometimes add or omit parentheses.

By $\free(\varphi)$ we denote the set of (node or set) variables
that occur free (i.e., not within the range of a node or set
quantifier) in $\varphi$.
A \emph{sentence} is a formula without free variables. 
We write $\varphi(x_1,\ldots,x_k,X_1,\ldots,X_\ell)$ to indicate that $\varphi$ has $k$
free node variables $x_1,\ldots,x_k$ and $\ell$ free set variables
$X_1,\ldots,X_\ell$.
For a $\tau$-structure $\A$, elements $a_1,\ldots,a_k\in A$, and sets
$A_1,\ldots,A_\ell\subseteq A$, we write
$\A\models\varphi(a_1,\ldots,a_k,A_1,\ldots,A_\ell)$ to indicate that
$\A$ satisfies the formula $\varphi$ when interpreting the free
occurrences of the variables $x_1,\ldots,x_k,X_1,\ldots,X_\ell$ with
$a_1,\ldots,a_k,A_1,\ldots,A_\ell$.
A formula $\varphi(x_1,\ldots,x_k)$ with $k$ free node variables and
no free set variable defines a $k$-ary query on $\tau$-structures which,
when evaluated in a $\tau$-structure $\A$, results in the $k$-ary relation
\begin{eqnarray*}
  \AF{\varphi}(\A) & \deff & 
  \setc{\ (a_1,\ldots,a_k)\in A^k\ }{\ \A \models \varphi(a_1,\ldots,a_k)\ }.
\end{eqnarray*}

\begin{Example}\label{Exa:MSO_child}
Consider the schema $\tau_u$ introduced in
Section~\ref{subsection:UnorderedTrees} for representing unordered
$\Sigma$-labeled trees for
$\Sigma=\set{\textit{Black},\textit{White}}$.
We present a unary $\FO(\tau_u)$-query $\varphi(x)$ such that for
every unordered $\Sigma$-labeled tree $T$ we have 

\begin{eqnarray*}
  \AF{\varphi}\big(\S_u(T)\big) 
 & = &
  \left\{
   \begin{array}{lp{6cm}}
     \set{\,\rroot^T\,} & if the root of $T$ has exactly two
       children labeled with the symbol \textit{White},
    \\[1ex]
     \emptyset & otherwise.
   \end{array}
  \right.
\end{eqnarray*}
To this end, we let $\varphi(x)$ be the following
$\MSO(\tau_u)$-formula:
\[
 \begin{array}{l}
   \nicht \exists u \,\Child(u,x) \ \ \und\\[0.5ex]
   \exists y \, \exists z\,
   \Big( \,
      y \neq  z \ \und \, 
      \Child(x,y) \, \und \, \Child(x,z) \, \und \,
      \Label_{\textit{White}}(y) \, \und \, 
      \Label_{\textit{White}}(z) \, \und \,
 \\
      \qquad \qquad 
      \forall v\, \big( \,
         \Child(x,v) \, \impl \, (\, v=y \, \oder \, v=z
         \, \oder \, \nicht\Label_{\textit{White}}(v) \,) 
      \, \big)	  
   \, \Big)
 \end{array} 	
\]
\markEnd
\end{Example}

\medskip

\noindent
A $\PiMSO(\tau)$-formula is an $\MSO(\tau)$-formula of the form 
\[ 
  \forall X_1 \cdots \forall X_m \ 
  \exists x_1 \cdots\exists x_k \ 
  \xi
\] 
where $m,k \in \NN$, $X_1,\ldots,X_m$ are set variables,
$x_1,\ldots,x_k$ are node variables, and $\xi$
is a formula that does not contain any (node or set) quantifier.

It is well-known that unary monadic datalog queries can be translated into
equivalent $\PiMSO$ queries.

\begin{Proposition}[Folklore]\label{prop:mDatalog2MSO}
 Let $\tau$ be a schema. For every unary monadic datalog query $Q=(\PP,P)$ of
 schema $\tau$ there is a $\PiMSO(\tau)$-formula $\varphi(x)$ such that
 $\AF{Q}(\A)=\AF{\varphi}(\A)$ is true
 for every finite $\tau$-structure $\A$.

 Furthermore, there is an algorithm which computes $\varphi$ from $Q$
 in time polynomial in the size of $Q$.
\end{Proposition}
\begin{proof}
Let $\{X_1, \ldots , X_m\} = \idb(\PP)$ be the set of intensional
predicates of $\PP$, and w.l.o.g let $X_1=P$.
For every rule $r$ of $\PP$ of the form
\ $  
   h \leftarrow b_1, \ldots , b_n,
$ \ 
with $\set{z_1,\ldots,z_k}=\Var(r)$ let 
\begin{eqnarray*}
  \psi_r(X_1,\ldots,X_m) & \deff &
  \forall z_1\, \cdots\, \forall z_k \ \big( \ 
     ( \; b_1 \und \cdots \und b_n \;)\; \impl\; h
  \ \big).
\end{eqnarray*}
Now, let
\ $
 \chi(X_1,\ldots,X_m) \deff 
 \Und_{r\in\PP} \psi_r(X_1,\ldots,X_m)$. \
Finally, let $x$ be a node variable that does not occur in
$\chi(X_1,\ldots,X_m)$ and let
\begin{eqnarray*}
  \varphi(x) & \deff &
  \forall X_1 \, \cdots \, \forall X_m \ \big( \
    \chi(X_1,\ldots,X_m)\;\impl\, X_1(x)
  \ \big).
\end{eqnarray*}

Obviously, $\varphi(x)$ is equivalent, on the
class of all $\tau$-structures, to the formula
\ $\forall X_1  \cdots  \forall X_m \, \big(
   X_1(x)  \oder  \nicht \chi \big)$, \
and $\nicht\chi$ is equivalent to $\Oder_{r\in\PP}\nicht\psi_r$, \ 
while $\nicht\psi_r$ is equivalent to 
\ $\exists z_1\cdots\exists z_k \,\nicht\big((b_1\und\cdots\und
b_n)\impl h\big)$. Thus, it is straightforward to see that
$\varphi(x)$ is equivalent to a $\PiMSO(\tau)$-formula, and this
formula can be constructed in time polynomial in the size of $Q$.

It remains to verify that $\AF{Q}(\A)=\AF{\varphi}(\A)$, for
every $\tau$-structure $\A$.
To this end, let $\A$ be an arbitrary $\tau$-structure.
By the construction of $\varphi(x)$ we know for $a\in A$ that
\begin{equation*}
 \begin{array}{ll}
 & a\in\AF{\varphi}(\A)
 \\[1ex]
   \iff
 & a\in X_1^{\A'} \ \ \text{for every
   $\tau\cup\set{X_1,\ldots,X_m}$-expansion $\A'$ of $\A$ 
   with $\A'\models\chi$}.
 \end{array}
\end{equation*}

Now let $C\deff\atoms(\A)$. Furthermore, consider arbitrary
sets $A_1,\ldots,A_m\subseteq A$, let 
$\A'$ be the $\tau\cup\set{X_1,\ldots,X_m}$-structure obtained as the
expansion of $\A$ by $X_i^\A\deff A_i$ for all
$i\in\set{1,\ldots,m}$, and let $D\deff\atoms(\A')$. Clearly,
$C\subseteq D\subseteq F_{\PP,A}$.
Furthermore, note that $\chi$ is constructed in such a way that the
following is true:
\[
 \A'\models \chi  \quad \iff \quad \T_\PP(D)\subseteq D.
\]
By the theorem of Knaster and Tarski (Theorem~\ref{thm:KnasterTarski})
we know that
\begin{eqnarray*}
  \T_\PP^\omega(C) & = &
  \bigcap \; \setc{\,D}{\,\T_{\PP}(D) \subseteq D \text{ \,and \,} C
   \subseteq D \subseteq F_{\PP,A}\,}.
\end{eqnarray*}
Thus, for $a\in A$ we have
\[
 \begin{array}{ll}
 &  a\in\AF{Q}(\A) 
 \\[1ex]
   \iff 
 & X_1(a)\in \T_\PP^\omega(C)
 \\[1ex]
   \iff
 & X_1(a)\in D \ \ \text{for every $D$ with $\T_\PP(D)\subseteq D$ and
   $C\subseteq D\subseteq F_{\PP,A}$}
 \\[1ex]
   \iff
 & a\in X_1^{\A'} \ \ \text{for every
   $\tau\cup\set{X_1,\ldots,X_m}$-expansion $\A'$ of $\A$ with
   $\A'\models\chi$}
 \\[1ex]
   \iff
 & a\in\AF{\varphi}(\A).
 \end{array}
\]
This completes the proof of Proposition~\ref{prop:mDatalog2MSO}.
\end{proof}

\section{Expressive Power of Monadic Datalog on Trees}\label{section:ExpressivePower}

A unary query $q$ on $\Sigma$-labeled (un)ordered trees assigns to each
(un)ordered $\Sigma$-labeled tree $T$ a set 
$q(T)\subseteq V^T$.

\subsection{Expressive Power of $\mDatalog$ on Ordered Trees}\label{subsection:ExpressivePower-OrderedTrees}

Let $\tau$ be one of the schemas introduced in
Section~\ref{subsection:OrderedTrees}, i.e., $\tau$ is $\tau_o^M$ for some 
$M\subseteq\set{\Child,\Desc,\Root,\Leaf,\Ls}$. 
We say that a unary query $q$ on $\Sigma$-labeled ordered trees is
\emph{$\mDatalog(\tau)$-definable} iff there is a unary monadic
datalog query $Q$ of schema $\tau$ such 
that for every ordered $\Sigma$-labeled tree $T$ we have
\ $q(T)  =  \AF{Q}(\S_o^M(T))$.
Similarly, for any subset $L$ of $\MSO$, $q$ is called
\emph{$L(\tau)$-definable} iff there is an $L(\tau)$-formula
$\varphi(x)$ such that for every ordered $\Sigma$-labeled tree $T$
we have
\ $q(T) =  \AF{\varphi}(\S_o^M(T))$.

Often, we will simply write $Q(T)$ instead of $\AF{Q}(\S_o^M(T))$, and
$\varphi(T)$ instead of $\AF{\varphi}(\S_o^M(T))$.

Proposition~\ref{prop:mDatalog2MSO} implies that unary queries
on $\Sigma$-labeled ordered trees
which are definable in $\mDatalog(\tau)$, are also definable in $\MSO(\tau)$.
In \cite{GottlobKoch} it was shown that for the particular schema
$\tau=\tauGK$ also the converse is true:

\begin{Theorem}[Gottlob, Koch \cite{GottlobKoch}]\label{thm:GottlobKoch}
 A unary query on $\Sigma$-labeled ordered trees is definable in
 $\mDatalog(\tauGK)$ if, and only if, it is definable in
 $\MSO(\tauGK)$.
\\
 Furthermore, there is an algorithm which translates a given unary
 $\mDatalog(\tauGK)$-query into an equivalent unary
 $\MSO(\tauGK)$-query, and vice versa.
\markEnd
\end{Theorem}

In the remainder of this subsection, we point out that adding the
$\Child$ and $\Desc$ relations won't increase the expressive power of $\mDatalog$ or
$\MSO$ on ordered $\Sigma$-labeled trees, while omitting any of the relations
$\Root$, $\Leaf$, or $\Ls$ will substantially decrease the expressive
power of $\mDatalog$, but not of $\MSO$.

\begin{Fact}[Folklore]\label{fact:MSO-orderedTrees}
There are $\MSO(\tau_o)$-formulas 
\begin{center}
$\varphi_{\Child}(x,y)$, \ 
$\varphi_{\Desc}(x,y)$, \ 
$\varphi_{\Root}(x)$, \ 
$\varphi_{\Leaf}(x)$, \ 
$\varphi_{\Ls}(x)$, 
\end{center}
such that for every $\Sigma$-labeled ordered tree $T$ and all nodes
$a,b$ of $T$ we have
\[
\begin{array}{lcl}
   \S_o(T)\models\varphi_{\Child}(a,b) & \iff & \S'_o(T)\models\Child(a,b),
\\
   \S_o(T)\models\varphi_{\Desc}(a,b) & \iff & \S'_o(T)\models\Desc(a,b),
\\
   \S_o(T)\models\varphi_{\Root}(a) & \iff & \S'_o(T)\models\Root(a),
\\
   \S_o(T)\models\varphi_{\Leaf}(a) & \iff & \S'_o(T)\models\Leaf(a),
\\
   \S_o(T)\models\varphi_{\Ls}(a) & \iff & \S'_o(T)\models\Ls(a).
\end{array}
\]
\end{Fact}
\begin{proof}
Obviously, we can choose 
\[
\begin{array}{rcl}
  \varphi_{\Root}(x) 
& \deff 
& \nicht\,\exists y\;\big(\,\Fc(y,x)\,\oder\,\Ns(y,x)\,\big),
\\[1ex]
  \varphi_{\Leaf}(x)
& \deff
& \nicht\,\exists y\; \Fc(x,y),
\\[1ex]
  \varphi_{\Ls}(x)
& \deff
& \nicht\,\exists y\; \Ns(x,y).
\end{array}
\]
For constructing $\varphi_{\Child}(x,y)$ and
$\varphi_{\Desc}(x,y)$, we consider the following auxiliary formulas:
Let $\varrho(x,y)$ be an arbitrary formula, let $X$ be a set variable,
and let
\begin{eqnarray*}
  \textit{cl}_{\varrho(x,y)}(X) 
& \deff
& \forall x\,\forall y\; \Big(\,
    \big(\,X(x) \und \varrho(x,y)\,\big) \,\impl\, X(y)
  \,\Big).
\end{eqnarray*}
Clearly, this formula holds for a set $X$ iff $X$ is closed under
``$\varrho$-successors''.

In particular, the formula
\begin{eqnarray*}
  \varphi_{\Ns^*}(x,y) 
& \deff
& \forall X\; \Big(\,
    \big(\, X(x) \,\und\, \textit{cl}_{\Ns(x,y)}(X) \,\big) \,\impl\, X(y)
  \,\Big)
\end{eqnarray*}
expresses that $y$ is either equal to $x$, or it is a sibling of $x$
which is bigger than $x$ w.r.t.\ the linear order of all children of
$x$ and $y$'s common parent.
Consequently, we can choose
\begin{eqnarray*}
  \varphi_{\Child}(x,y)
& \deff
& \exists x'\;
    \big(\,   
      \Fc(x,x') \,\und\, \varphi_{\Ns^*}(x',y)
    \,\big).
\end{eqnarray*}
Since the $\Desc$-relation is the transitive (and non-reflexive)
closure of the $\Child$-relation, we can choose
\begin{eqnarray*}
  \varphi_{\Desc}(x,y)
& \deff
& x\neq y \ \und \ 
  \forall X\; \Big(\,
    \big(\,
      X(x) \,\und\,\textit{cl}_{\varphi_{\Child}(x,y)}(X)
    \,\big) \,\impl\, X(y)
  \,\Big).
\end{eqnarray*}
\end{proof}

\noindent
In combination with Theorem~\ref{thm:GottlobKoch} and Proposition~\ref{prop:mDatalog2MSO}, this leads to:

\begin{Corollary}\label{cor:mDatalogVsMSO-orderedTrees}
 The following languages can express exactly the same unary queries on $\Sigma$-labeled ordered trees:
 \begin{center}
   $\mDatalog(\tauGK)$, \ $\mDatalog(\tau'_o)$, \ 
   $\MSO(\tau'_o)$, \ $\MSO(\tauGK)$, \ $\MSO(\tau_o)$.
 \end{center}
Furthermore, there is an algorithm which translates a given unary
query on $\Sigma$-labeled ordered trees formulated in one of these
languages into equivalent queries formulated in any of the other languages.

In particular, adding the $\Child$ and $\Desc$ relations to $\tauGK$
does not increase the expressive power of monadic datalog on
$\Sigma$-labeled ordered trees.
\end{Corollary}
\begin{proof}
Since $\tauGK\subseteq\tau'_o$, \ $\mDatalog(\tauGK)$ is at most as
expressive as $\mDatalog(\tau'_o)$ which, by
Proposition~\ref{prop:mDatalog2MSO}, is at most as expressive as
$\MSO(\tau'_o)$. 

By Fact~\ref{fact:MSO-orderedTrees}, $\MSO(\tau'_o)$ is as expressive
on $\Sigma$-labeled ordered trees as $\MSO(\tau_o)$ and $\MSO(\tauGK)$ which, by
Theorem~\ref{thm:GottlobKoch}, is as expressive on $\Sigma$-labeled
ordered trees as $\mDatalog(\tauGK)$.

Furthermore, by Proposition~\ref{prop:mDatalog2MSO},
Fact~\ref{fact:MSO-orderedTrees}, and Theorem~\ref{thm:GottlobKoch},
the
translation from one language to another is constructive.
\end{proof}

Next, we note that omitting any of the unary relations $\Root$,
$\Leaf$, or $\Ls$ decreases the expressive power of monadic datalog on
$\Sigma$-labeled ordered trees.

\begin{Observation}\label{obs:mDatalogWithoutUnaryRel-orderedTrees}
For any relation $\Rel\in\set{\Root,\Leaf,\Ls}$, the unary query 
$q_{\Rel}$ with $q_{\Rel}(T)=\setc{v\in V^T}{\S'_o(T)\models\Rel(v)}$
can be expressed in $\mDatalog(\set{\Rel})$, but not in 
$\mDatalog(\tau'_o\setminus\set{\Rel})$.
\end{Observation}
\begin{proof}
 It is obvious that the query $q_{\Rel}$ can be expressed in
 $\mDatalog(\set{\Rel})$.

 Let $M\subseteq\set{\Child,\Desc,\Root,\Leaf,\Ls}$ be such that 
 $\tau_o^M=\tau'_o\setminus\set{\Rel}$.
 Assume, for contradiction, that $q_{\Rel}$ is expressed by a
 $\mDatalog(\tau_o^M)$-query $Q=(\PP,P)$.

 First, consider the case where $\Rel=\Root$. Let $T_0$ be the tree
 consisting of a single node $v$ labeled $\alpha\in\Sigma$, and let
 $T_1$ be the tree consisting of two nodes $u,v$, both labeled
 $\alpha$, such that $v$ is the unique child of $u$.
 Since $\tau_o^M=\tau'_o\setminus\set{\Root}$, we have
 \[
 \begin{array}{rl}
   \atoms\big(\S_o^M(T_0)\big) 
 \ = 
 & \{\ 
      \Label_\alpha(v),\ \Leaf(v)
   \ \}, \quad \text{and}
 \\[1ex]
   \atoms\big(\S_o^M(T_1)\big)
 \ =
 & \atoms\big(\S_o^M(T_0)\big) \ \cup\ 
  \left\{
   \begin{array}{l}
     \Label_\alpha(u), \ \Fc(u,v), \\ \Ls(v),  \ \Child(u,v),\\
     \Desc(u,v)
   \end{array}
   \right\}.
 \end{array}
 \]
 I.e., $\atoms(\S_o^M(T_0))\subseteq \atoms(\S_o^M(T_1))$ and thus,
 due to the monotonicity stated in Remark~\ref{remark:monotonicity}, we have 
 $\AF{Q}(\S_o^M(T_0))\subseteq \AF{Q}(\S_o^M(T_1))$. 
 This contradicts the fact that 
 $v\in q_{\Root}(T_0)=\AF{Q}(\S_o^M(T_0))$ but $v\not\in
 q_{\Root}(T_1)=\AF{Q}(\S_o^M(T_1))$.

 Next, consider the case where $\Rel=\Leaf$, and let
 $T_0$ be the tree consisting of a single node $v$ labeled $\alpha\in\Sigma$, and let
 $T'_1$ be the tree consisting of two nodes $v$ and $w$, both labeled
 $\alpha$, such that $w$ is the unique child of $v$.
 Since $\tau_o^M=\tau'_o\setminus\set{\Leaf}$, it is straightforward
 to see that $\atoms(\S_o^M(T_0))\subseteq\atoms(\S_o^M(T'_1))$.
 By monotonicity, we have that 
 $\AF{Q}(\S_o^M(T_0))\subseteq \AF{Q}(\S_o^M(T'_1))$, 
 contradicting the fact that 
 $v\in q_{\Leaf}(T_0)=\AF{Q}(\S_o^M(T_0))$ but $v\not\in
 q_{\Leaf}(T'_1)=\AF{Q}(\S_o^M(T'_1))$.

Finally, consider the case where $\Rel=\Ls$.
Let $T_1$ be the tree consisting of two nodes $u,v$, both labeled
$\alpha$, such that $v$ is the unique child of $u$. Let
$T_2$ be the tree consisting of three nodes $u,v,w$, all labeled
$\alpha$, such that $v$ and $w$ are the first and the second child of
$u$. 
Since $\tau_o^M=\tau'_o\setminus\set{\Ls}$, it is straightforward
to see that $\atoms(\S_o^M(T_1))\subseteq\atoms(\S_o^M(T_2))$.
By monotonicity, we have
$\AF{Q}(\S_o^M(T_1))\subseteq \AF{Q}(\S_o^M(T_2))$, 
contradicting the fact that 
 $v\in q_{\Ls}(T_1)=\AF{Q}(\S_o^M(T_1))$ but $v\not\in
 q_{\Ls}(T_2)=\AF{Q}(\S_o^M(T_2))$.
\end{proof}

\subsection{Expressive Power of $\mDatalog$ on Unordered Trees}\label{subsection:ExpressivePower-UnorderedTrees}

Let $\tau$ be one of the schemas introduced in
Section~\ref{subsection:UnorderedTrees}, i.e., $\tau$ is $\tau_u^M$ for some 
$M\subseteq\set{\Desc,\Is,\Root,\Leaf}$. 
We say that a unary query $q$ on $\Sigma$-labeled unordered trees is
\emph{$\mDatalog(\tau)$-definable} iff there is a unary monadic
datalog query $Q$ of schema $\tau$ such 
that for every unordered $\Sigma$-labeled tree $T$ we have
\ $q(T)  =  \AF{Q}(\S_u^M(T))$.
Similarly, for any subset $L$ of $\MSO$, $q$ is called
\emph{$L(\tau)$-definable} iff there is an $L(\tau)$-formula
$\varphi(x)$ such that for every unordered $\Sigma$-labeled tree $T$
we have
\ $q(T) =  \AF{\varphi}(\S_u^M(T))$.

Often, we will simply write $Q(T)$ instead of $\AF{Q}(\S_u^M(T))$, and
$\varphi(T)$ instead of $\AF{\varphi}(\S_u^M(T))$.

Proposition~\ref{prop:mDatalog2MSO} implies that unary queries
on $\Sigma$-labeled unordered trees
which are definable in $\mDatalog(\tau)$, are also definable in $\MSO(\tau)$.
It is straightforward to see that
 $\MSO(\tau_u)$ can express all the relations present in $\tau'_u$:

\begin{Fact}[Folklore]\label{fact:MSO-unorderedTrees}
 There are $\MSO(\tau_u)$-formulas 
 \[
   \varphi_{\Desc}(x,y), \ \
   \varphi_{\As}(x,y),\ \
   \varphi_{\Root}(x),\ \
   \varphi_{\Leaf}(x)
 \]
 such that for every $\Sigma$-labeled unordered tree $T$ and all nodes
 $a,b$ of $T$ we have
\[
 \begin{array}{lcl}
  \S_u(T)\models\varphi_{\Desc}(a,b) & \iff &
  \S'_u(T)\models\Desc(a,b),
\\
  \S_u(T)\models\varphi_{\As}(a,b) & \iff & \S'_u(T)\models\As(a,b),
\\
  \S_u(T)\models\varphi_{\Root}(a) & \iff & \S'_u(T)\models\Root(a),
\\
  \S_u(T)\models\varphi_{\Leaf}(a) & \iff & \S'_u(T)\models\Leaf(a).
 \end{array}
\]
\end{Fact}
\begin{proof}
Obviously, we can choose
\begin{eqnarray*}
  \varphi_{\Root}(x) 
& \deff
& \nicht\,\exists y\,\Child(y,x),
\\
  \varphi_{\Leaf}(x)
& \deff
& \nicht\,\exists y\,\Child(x,y),
\\
  \varphi_{\As}(x,y)
& \deff
& x\neq y \ \und \ \exists u\,\big(\, \Child(u,x)\,\und\,\Child(u,y)\,\big).
\end{eqnarray*}
For constructing 
$\varphi_{\Desc}(x,y)$, we consider the following auxiliary formula:
Let $\varrho(x,y)$ be an arbitrary formula, let $X$ be a set variable,
and let
\begin{eqnarray*}
  \textit{cl}_{\varrho(x,y)}(X) 
& \deff
& \forall x\,\forall y\; \Big(\,
    \big(\,X(x) \und \varrho(x,y)\,\big) \,\impl\, X(y)
  \,\Big).
\end{eqnarray*}
Clearly, this formula holds for a set $X$ iff $X$ is closed under
``$\varrho$-successors''.

In particular, the formula
\begin{eqnarray*}
  \varphi_{\Child^*}(x,y) 
& \deff
& \forall X\; \Big(\,
    \big(\, X(x) \,\und\, \textit{cl}_{\Child(x,y)}(X) \,\big) \,\impl\, X(y)
  \,\Big)
\end{eqnarray*}
expresses that $y$ is either equal to $x$, or it is a 
descendant of $x$. 
Thus, we can choose
\begin{eqnarray*}
  \varphi_{\Desc}(x,y)
& \deff
& x\neq y \ \und \ \varphi_{\Child^*}(x,y).
\end{eqnarray*}
\end{proof}

However, unlike in the case of ordered trees, $\mDatalog(\tau'_u)$
cannot express all unary queries expressible in $\MSO(\tau_u)$, as the
following observation shows.

\begin{Observation}\label{obs:mDatalog-weaker-than-MSO-unorderedTrees}
The unary query $q_{\textit{two}}$ with 
\[
  q_\textit{two}(T) \ \ = \ \
  \setc{v\in V^T}{v \text{ has exactly two children labeled }\alpha}
 \]
is expressible in
$\MSO(\tau_u)$, but not in $\mDatalog(\tau'_u)$.
\end{Observation}
\begin{proof}
It is obvious that the query $q_\textit{two}$ is defined by the
$\MSO(\tau_u)$-formula \ $\psi(x) \deff$
\[
\begin{array}{ll}
  \exists y_1\,\exists y_2\, \Big( \hspace{-2ex}
& \Child(x,y_1)\,\und\,\Child(x,y_2)\,\und\, \Label_\alpha(y_1)
  \,\und\, \Label_\alpha(y_2)\,\und\, y_1\neq y_2 \, \und
\\
& \forall z\,\big(\,
   (\, \Child(x,z)\,\und\,\Label_\alpha(z)\,)\,\impl\,
   (\, z=y_1\,\oder\, z=y_2\,)\,
  \big)\,
  \Big).
\end{array}
\]

For contradiction, assume that $q_\textit{two}$ is expressed by a $\mDatalog(\tau'_u)$-query
$Q=(\PP,P)$.
Let $T_2$ be the $\Sigma$-labeled unordered tree consisting of three
nodes $u,v_1,v_2$, all labeled $\alpha$, such that $v_1$ and $v_2$ are
children of $u$. Furthermore, let $T_3$ be the tree consisting of four
nodes $u,v_1,v_2,v_3$, all labeled $\alpha$, such that $v_1,v_2,v_3$ are children of $u$.
Since 
\[
  \tau'_u
  \ = \ 
  \setc{\Label_\alpha}{\alpha\in \Sigma}\ \cup\
  \set{\,\Child,\,\Desc,\,\As,\,\Root,\,\Leaf\,},
\]
it is straightforward to see that
$\atoms(\S'_u(T_2))\subseteq\atoms(\S'_u(T_3))$.
Thus, due to the monotonicity stated in
Remark~\ref{remark:monotonicity}, we have 
$\AF{Q}(\S'_u(T_2))\subseteq\AF{Q}(\S'_u(T_3))$. 
This contradicts the fact that $u\in
q_\textit{two}(T_2)=\AF{Q}(\S'_u(T_2))$ but 
$u\not\in q_{\textit{two}}(T_3)=\AF{Q}(\S'_u(T_3))$.
\end{proof}

Next, we note that omitting any of the relations $\Root$, $\Leaf$, or $\As$
further decreases the expressive power of monadic datalog on
$\Sigma$-labeled unordered trees.

\begin{Observation}\label{obs:mDatalogWithoutUnaryRel-unorderedTrees}
\begin{mea}
 \item
 For any relation $\Rel\in\set{\Root,\Leaf}$, the query 
 $q_{\Rel}$ with $q_{\Rel}(T) = \setc{v\in
   V^T}{\S'_u(T)\models\Rel(v)}$ can be expressed in
 $\mDatalog(\set{\Rel})$, but not in
 $\mDatalog(\tau'_u\setminus{\Rel})$.
 \item
 The query $q_{\textit{sib}}$ with $q_{\textit{sib}}(T) =
 \setc{v\in V^T}{v\text{ has at least one sibling}}$, for all
 $\Sigma$-labeled unordered trees $T$, can be expressed
 in $\mDatalog(\set{\As})$, but not in
 $\mDatalog(\tau'_u\setminus\set{\As})$.
\end{mea}
\end{Observation}
\begin{proof}
The proof of (a) is analogous to the according parts of the proof of
Observation~\ref{obs:mDatalogWithoutUnaryRel-orderedTrees}. 

For the proof of (b), first note that $q_{\textit{sib}}$ is
expressed by the unary monadic datalog query $Q=(\PP,P)$ where $Q$
consists of the single rule 
\[
  P(x)\leftarrow \As(x,y).
\]
Now let $M= \set{\Desc,\Root,\Leaf}$, i.e.,
$\tau_u^M=\tau'_u\setminus\set{\As}$.
Assume, for contradiction, that $q_{\textit{sib}}$ is expressed by a
unary $\mDatalog(\tau_u^M)$-query $Q=(\PP,P)$.
We will conclude the proof by using 
Lemma~\ref{lemma:homomorphisms},
stating that datalog queries are preserved under
homomorphisms.

Let $T_2$ be the $\Sigma$-labeled unordered tree consisting of three
nodes $a,a_1,a_2$, all labeled $\alpha$, such that $a_1$ and $a_2$ are 
children of $a$. Furthermore, let $T_1$ be the tree consisting of two nodes
$b,b_1$, both labeled $\alpha$, such that $b_1$ is the unique child of $b$.
Let $\A\deff \S_u^M(T_2)$ and $\B\deff \S_u^M(T_1)$.

Consider the mapping $h:A\to B$ with $h(a)=b$ and
$h(a_1)=h(a_2)=b_1$. It is not difficult to see that $h$ is a
homomorphism from $\A$ to $\B$, since
\begin{mi}
 \item $\Label_\alpha^\A=\set{a,a_1,a_2}$ \ and \
   $\Label_\alpha^\B=\set{b,b_1}$
 \item $\Label_{\alpha'}^\A=\emptyset=\Label_{\alpha'}^\B$, \ for all
   $\alpha'\in\Sigma$ with $\alpha'\neq \alpha$
 \item $\Child^\A=\set{(a,a_1),\, (a,a_2)}$ \ and \
  $\Child^\B=\set{(b,b_1)}$
 \item $\Desc^\A=\Child^\A$ \ and \ $\Desc^B=\Child^B$
 \item $\Root^\A=\set{a}$ \ and \ $\Root^\B=\set{b}$
 \item $\Leaf^\A=\set{a_1,\,a_2}$ \ and \ $\Leaf^\B=\set{b_1}$.
\end{mi}
From Lemma~\ref{lemma:homomorphisms} we obtain that
$h\big(\AF{Q}(\A)\big)\subseteq \AF{Q}(\B)$.
This contradicts the fact that $a_1\in
q_{\textit{sib}}(T_2)=\AF{Q}(\A)$, but $h(a_1)=b_1\not\in q_{\textit{sib}}(T_1)=\AF{Q}(\B)$.
\end{proof}

\noindent
In summary, we immediately obtain the following:

\begin{Corollary}\label{cor:ExpressivePower-unorderedTrees}
 \begin{mea}
  \item
   $\MSO(\tau_u)$ can express exactly the same unary queries on
   $\Sigma$-labeled unordered trees as $\MSO(\tau'_u)$;
   and there is a polynomial time algorithm which translates a given unary
   $\MSO(\tau'_u)$-query on $\Sigma$-labeled unordered trees into an equivalent
   $\MSO(\tau_u)$-query.

   Furthermore, both languages
   are capable of expressing strictly more unary queries on
   $\Sigma$-labeled unordered trees than $\mDatalog(\tau'_u)$.
  \item
   Omitting any of the relations $\Root$, $\Leaf$, $\As$ strictly
   decreases the expressive power of unary
   $\mDatalog(\tau'_u)$-queries on $\Sigma$-labeled unordered trees.
\markEnd
 \end{mea}
\end{Corollary}

\section{Query containment, Equivalence, and Satisfiability for
  Monadic Datalog  on Trees}\label{section:QCPofMonadicDatalog}

Query containment, equivalence, and satisfiability of queries are important
problems concerning query optimisation and static analysis of queries.

\subsection{Query Containment for $\mDatalog$ on Trees}

Let $\tau$ be one of the schemas introduced in
Section~\ref{subsection:UnorderedTrees} or
Section~\ref{subsection:OrderedTrees}, and let $\S(T)$ the corresponding
$\tau$-structure representing the tree $T$. 

For two queries $Q_1$ and $Q_2$ of schema $\tau$, 
we write $Q_1 \subseteq Q_2$ (and say that $Q_1$ is included in $Q_2$
on trees) to indicate that for every  $\Sigma$-labeled
tree $T$ we have $\AF{Q_1}(\S(T)) \subseteq \AF{Q_2}(\S(T))$.
Accordingly, we write $Q_1\not\subseteq Q_2$ to indicate that
$Q_1\subseteq Q_2$ does not hold. 

An important task for query optimisation and static analysis 
is the \emph{query containment problem}, defined
as follows:

\begin{Problem}{The QCP for unary $\mDatalog(\tau)$-queries on (un)ordered
    $\Sigma$-labeled trees}
	\In Two unary $\mDatalog(\tau)$-queries $Q_1$ and $Q_2$. 
	\Out 
          \begin{tabular}[t]{lp{7cm}}
          \Yes, & if $Q_1\subseteq Q_2$,
          \\
	  \No,  & otherwise.
          \end{tabular} 
\end{Problem}

\noindent
For \emph{ordered} $\Sigma$-labeled trees, the following is known:

\begin{Theorem}[Gottlob, Koch \cite{GottlobKoch}]\label{QCP:ur} \ \\
 The QCP for unary $\mDatalog(\tauGK)$-queries on ordered $\Sigma$-labeled trees is
 decidable and $\EXPTIME$-hard.
\markEnd
\end{Theorem}

Using Corollary~\ref{cor:mDatalogVsMSO-orderedTrees} and the fact that
$\tauGK\subseteq \tau'_o$, this immediately
leads to:

\begin{Theorem}\label{thm:QCP_orderedTrees}
 The QCP for unary $\mDatalog(\tau'_o)$-queries on ordered $\Sigma$-labeled trees is
 decidable and $\EXPTIME$-hard.
\markEnd
\end{Theorem}

To obtain decidability also for the case of \emph{unordered}
$\Sigma$-labeled trees, we can use the following result:

\begin{Theorem}[Seese \cite{Seese91}]\label{MSO:SAT}
 The problem
 \begin{Problem}{Satisfiability of $\MSO(\tau_{u})$-sentences on
     unordered $\Sigma$-labeled trees}
 	\In An $\MSO(\tau_{u})$-sentence $\varphi$. 
 	\Quest Does there exist an unordered $\Sigma$-labeled (finite) tree $T$ such that $\S_u(T) \models \varphi$?
 \end{Problem}
 is decidable.\markEnd
\end{Theorem}

Combining this with Proposition~\ref{prop:mDatalog2MSO} and
Fact~\ref{fact:MSO-unorderedTrees}, we obtain:

\begin{Theorem} \label{thm:QCP_unordered}
 The QCP for unary $\mDatalog(\tau'_u)$-queries on unordered $\Sigma$-labeled trees is decidable.
\end{Theorem}
\begin{proof}
An algorithm for deciding the QCP for unary $\mDatalog(\tau'_u)$-queries on
unordered $\Sigma$-labeled trees can proceed as follows:

On input of two unary $\mDatalog(\tau'_u)$-queries $Q_1$ and $Q_2$, first
use the algorithm from Proposition~\ref{prop:mDatalog2MSO} to
construct two $\MSO(\tau'_u)$-formulas $\varphi_1(x)$ and $\varphi_2(x)$ such that, for
each $i\in\set{1,2}$, the formula $\varphi_i(x)$ defines
the same unary query on $\Sigma$-labeled unordered trees as $Q_i$.

Afterwards, use Fact~\ref{fact:MSO-unorderedTrees} to translate the $\MSO(\tau'_u)$-formulas
$\varphi_1(x)$ and $\varphi_2(x)$ into $\MSO(\tau_u)$-formulas $\psi_1(x)$ and
$\psi_2(x)$, which are equivalent to $\varphi_1(x)$ and $\varphi_2(x)$
on $\Sigma$-labeled unordered trees.

Finally, let 
\[
  \varphi \quad \deff \quad
  \exists x\ \big(\, \psi_1(x) \, \und \, \nicht\psi_2(x)\,\big),
\]
and use the algorithm provided by Theorem~\ref{MSO:SAT} to decide
whether there is an unordered $\Sigma$-labeled tree $T$ such that
$\S_u(T)\models\varphi$.
Output ``\No'' if this algorithm outputs ``\Yes'', and output ``\Yes'' otherwise.

To verify that this algorithm produces the correct answer, note that
for every $\Sigma$-labeled unordered tree $T$, the following
is true:
\[
 \begin{array}{ll}
& \S_u(T)\models\varphi
\\
  \iff
& \text{there is a node $a$ of $T$ with
  $\S'_u(T)\models\psi_1(a)$ and $\S'_u(T)\not\models\psi_2(a)$}
\\
  \iff
& \text{there is a node $a$ of $T$ with
  $a\in\AF{Q_1}(\S'_u(T))$ and $a\not\in\AF{Q_2}(\S'_u(T))$}
\\
  \iff
& \AF{Q_1}(\S'_u(T))\not\subseteq \AF{Q_2}(\S'_u(T)).
 \end{array}
\]
Thus, the $\MSO(\tau_u)$-sentence $\varphi$ is satisfiable on
unordered $\Sigma$-labeled trees if, and only if, $Q_1\not\subseteq
Q_2$. 
\end{proof}

\subsection{Equivalence for $\mDatalog$ on Trees}

Let $\tau$ be one of the schemas introduced in
Section~\ref{subsection:UnorderedTrees} or
Section~\ref{subsection:OrderedTrees}, and let $\S(T)$ the corresponding
$\tau$-structure representing the tree $T$. 

For two queries $Q_1$ and $Q_2$ of schema $\tau$, 
we write $Q_1\equiv Q_2$ (and say that $Q_1$ is equivalent to $Q_2$ on
trees)
to indicate that for every $\Sigma$-labeled tree $T$ we have
$\AF{Q_1}(\S(T))=\AF{Q_2}(\S(T))$.
Accordingly, we write $Q_1\not\equiv Q_2$ to indicate that $Q_1\equiv
Q_2$ does not hold.
We consider the following decision problem.

\begin{Problem}{The Equivalence Problem for unary
    $\mDatalog(\tau)$-queries on $\Sigma$-labeled (un)ordered trees}
	\In Two unary $\mDatalog(\tau)$-queries $Q_1$ and $Q_2$.
        \Out
          \begin{tabular}[t]{lp{7.5cm}}
          \Yes, & if $Q_1\equiv Q_2$,
          \\
	  \No,  & otherwise.
          \end{tabular} 
\end{Problem}

By definition, we have $Q_1\equiv Q_2$ if, and only if, $Q_1\subseteq
Q_2$ and $Q_2\subseteq Q_1$. Thus, the decidability of the query
containment  problem for $\mDatalog$ stated in Theorem~\ref{thm:QCP_orderedTrees} and 
Theorem~\ref{thm:QCP_unordered}, immediately leads to the following.

\begin{Corollary}\label{cor:equivalence} \hspace{4cm}
\begin{mea}
 \item
   The equivalence problem for unary $\mDatalog(\tau'_o)$-queries on
   $\Sigma$-labeled ordered trees is decidable.
 \item
   The equivalence problem for unary $\mDatalog(\tau'_u)$-queries on
   $\Sigma$-labeled unordered trees is decidable.
 \markEnd
\end{mea}
\end{Corollary}

\subsection{Satisfiability of $\mDatalog$ on Trees}

Let $\tau$ be one of the schemas introduced in
Section~\ref{subsection:UnorderedTrees} or
Section~\ref{subsection:OrderedTrees}, and let $\S(T)$ the corresponding
$\tau$-structure representing the tree $T$. 

A query $Q$ of schema $\tau$ is called \emph{satisfiable on trees} if there
is a $\Sigma$-labeled (un)ordered tree $T$ such that $\AF{Q}(\S(T))\neq \emptyset$.

\begin{Example}\label{Exa:unsatisQ} There exists a unary $\mDatalog(\tau)$-query
 	$Q_\textit{unsat}=(\PP_\textit{unsat},P_\textit{unsat})$
 	which is \emph{not} satisfiable on trees. \\
	 For example, for $\tau=\tau_u$ the $\PP_\textit{unsat}$ can be chosen to consist of the
 	single rule
 	\[
   	P_\textit{unsat}(x) \ \ \leftarrow \ \ \Child(x,x)
 	\]
 	and for $\tau=\tau_o$ the following rule can be chosen
 	\[
   	P_\textit{unsat}(x) \ \ \leftarrow \ \ \Fc(x,x)
 	\]
	since in trees no node can be its own (first)child.
 \markEnd
\end{Example}

\noindent
We consider the following decision problem.
 
 \begin{Problem}{The Satisfiability Problem for unary
     $\mDatalog(\tau)$-queries on $\Sigma$-labeled (un)ordered trees}
 \label{mDatalog:SAT}
	\In A unary $\mDatalog(\tau)$-query $Q$.
        \Out
          \begin{tabular}[t]{lp{8cm}}
          \Yes, & if $Q$ is satisfiable on trees,
          \\
	  \No,  & otherwise.
          \end{tabular} 
\end{Problem}

\medskip

Corollary~\ref{cor:equivalence}, together with
Example~\ref{Exa:unsatisQ}, leads to the following.

\begin{Corollary} \label{cor:SAT}
 \begin{mea}
 \item The satisfiability problem for unary
   $\mDatalog(\tau'_o)$-queries on $\Sigma$-labeled ordered trees is decidable.
 \item The satisfiability problem for unary
   $\mDatalog(\tau'_u)$-queries on $\Sigma$-labeled unordered trees is decidable.
 \end{mea}	
\end{Corollary}
\begin{proof} 
Let $Q$ be the input query for which we want to decide whether or not
it is satisfiable on trees.
Let $Q_\textit{unsat}$ be the unsatisfiable query from
Example~\ref{Exa:unsatisQ}.

It is straightforward to see that $Q\equiv Q_\textit{unsat}$ if, and
only if, $Q$ is \emph{not} satisfiable on trees.
Thus, we can use the algorithms for deciding equivalence of queries on
trees (provided by Corollary~\ref{cor:equivalence}) to decide whether or not
$Q$ is satisfiable on trees.
\end{proof}

\bibliographystyle{amsplain}
\bibliography{a_note_on_expressive_power}

\providecommand{\bysame}{\leavevmode\hbox to3em{\hrulefill}\thinspace}
\providecommand{\MR}{\relax\ifhmode\unskip\space\fi MR }
\providecommand{\MRhref}[2]{%
  \href{http://www.ams.org/mathscinet-getitem?mr=#1}{#2}
}
\providecommand{\href}[2]{#2}
\begin{thebibliography}{10}

\bibitem{AbiteboulBMW13}
Serge Abiteboul, Pierre Bourhis, Anca Muscholl, and Zhilin Wu, \emph{Recursive
  queries on trees and data trees}, Proceedings of the 16th International
  Conference on Database Theory (ICDT'13), ACM, 2013, pp.~93--104.

\bibitem{DBLP:journals/jcss/BjorklundMS11}
Henrik Bj{\"o}rklund, Wim Martens, and Thomas Schwentick, \emph{Conjunctive
  query containment over trees}, J. Comput. Syst. Sci. \textbf{77} (2011),
  no.~3, 450--472.

\bibitem{DBLP:journals/jacm/BojanczykMSS09}
Mikolaj Bojanczyk, Anca Muscholl, Thomas Schwentick, and Luc Segoufin,
  \emph{Two-variable logic on data trees and xml reasoning}, J. ACM \textbf{56}
  (2009), no.~3.

\bibitem{tata2008}
H.~Comon, M.~Dauchet, R.~Gilleron, C.~L\"oding, F.~Jacquemard, D.~Lugiez,
  S.~Tison, and M.~Tommasi, \emph{Tree automata techniques and applications},
  Available on: \url{http://www.grappa.univ-lille3.fr/tata}, 2008, release
  November, 18th 2008.

\bibitem{GottlobKoch}
G.~Gottlob and C.~Koch, \emph{Monadic datalog and the expressive power of
  languages for web information extraction}, J. ACM \textbf{51} (2004), no.~1,
  pp. 74--113.

\bibitem{Kepser2008}
Stephan Kepser, \emph{A landscape of logics for finite unordered unranked
  trees}, Formal Grammar 2008 (Philippe de~Groote, Laura Kallmeyer, Gerald
  Penn, and Giorgio Satta, eds.), CSLI Publications, 2008.

\bibitem{Libkin}
Leonid Libkin, \emph{Elements of finite model theory}, Springer-Verlag, 2004.

\bibitem{DBLP:journals/lmcs/Libkin06}
\bysame, \emph{Logics for unranked trees: An overview}, Logical Methods in
  Computer Science \textbf{2} (2006), no.~3.

\bibitem{DBLP:conf/csl/Neven02}
Frank Neven, \emph{Automata, logic, and xml}, Proc.\ 16th International
  Workshop, CSL 2002, 11th Annual Conference of the EACSL (CSL'02), Lecture
  Notes in Computer Science, vol. 2471, Springer-Verlag, 2002, pp.~2--26.

\bibitem{DBLP:journals/tcs/NevenS02}
Frank Neven and Thomas Schwentick, \emph{Query automata over finite trees},
  Theor. Comput. Sci. \textbf{275} (2002), no.~1-2, 633--674.

\bibitem{Seese91}
D.~Seese, \emph{The structure of the models of decidable monadic theories of
  graphs}, Annals of Pure and Applied Logic \textbf{53} (1991), no.~2,
  169--195.

\bibitem{Tar}
Alfred Tarski, \emph{A lattice-theoretical fixpoint theorem and its
  applications}, Pacific Journal of Mathematics \textbf{5} (1955), no.~2,
  285--309.

\bibitem{WThomas-Handbook-survey}
Wolfgang Thomas, \emph{Languages, automata, and logic}, Handbook of Formal
  Languages, vol.~3, Springer-Verlag, 1997, pp.~389--455.

\end{thebibliography}

\end{document}